\documentclass[conference]{IEEEtran}
\IEEEoverridecommandlockouts

\usepackage{graphicx} 
\graphicspath{ {./Images/} }
\usepackage{amssymb}
\usepackage{latexsym}
\usepackage{soul}  
\usepackage{xcolor}  

\usepackage{mathtools, cuted}
\usepackage{lipsum, color}
\usepackage[font=small]{caption}

\usepackage{amsmath}

\usepackage{listings}
\usepackage{algorithmic}
\usepackage{textgreek}
\usepackage{amsthm}
\usepackage{amssymb}
\usepackage{mdframed}
\usepackage[ruled,vlined]{algorithm2e}[H]
\usepackage{hyperref}
\theoremstyle{definition}
\newcommand{\comment}[1]{}

\newtheorem{lemma}{Lemma}

\newtheorem{definition}{Definition}

\newtheorem{proposition}{Proposition}

\usepackage{mdframed}

\newmdtheoremenv{problem_stmt}{Problem}

\usepackage[belowskip=-15pt,aboveskip=0pt]{caption}

\usepackage{cite}
\usepackage{amsmath,amssymb,amsfonts}
\usepackage{algorithmic}
\usepackage{graphicx}
\usepackage{textcomp}
\usepackage{xcolor}
\def\BibTeX{{\rm B\kern-.05em{\sc i\kern-.025em b}\kern-.08em
    T\kern-.1667em\lower.7ex\hbox{E}\kern-.125emX}}
\begin{document}
\bstctlcite{IEEEexample:BSTcontrol}
\title{On Learning Spatial Provenance in Privacy-Constrained Wireless Networks\\

}



\author{\IEEEauthorblockN{Manish Bansal, \space 
  Pramsu Shrivastava and \space   J. Harshan}
Indian Institute of Technology Delhi, India\\

}
\maketitle

\begin{abstract}
In Vehicle-to-Everything networks that involve multi-hop communication, the Road Side Units (RSUs) typically aim to collect location information from the participating vehicles to provide security and network diagnostics features. While the vehicles commonly use the Global Positioning System (GPS) for navigation, they may refrain from sharing their precise GPS coordinates with the RSUs due to privacy concerns. Therefore, to jointly address the high localization requirements by the RSUs as well as the vehicles' privacy, we present a novel spatial-provenance framework wherein each vehicle uses Bloom filters to embed their partial location information when forwarding the packets. In this framework, the RSUs and the vehicles agree upon fragmenting the coverage area into several smaller regions so that the vehicles can embed the identity of their regions through Bloom filters. Given the probabilistic nature of Bloom filters, we derive an analytical expression on the error-rates in provenance recovery and then pose an optimization problem to choose the underlying parameters. With the help of extensive simulation results, we show that our method offers near-optimal Bloom filter parameters in learning spatial provenance. Some interesting trade-offs between the communication-overhead, spatial privacy of the vehicles and the error rates in provenance recovery are also discussed. 

\end{abstract}

\section{Introduction}
Vehicle-to-Everything (V2X) networks are expected to play a vital role in smart mobility as they enhance road safety and traffic efficiency \cite{v2x}. In V2X networks inter-vehicle communication (V2V) and vehicle-to-infrastructure (V2I) communication is facilitated via Road Side Units (RSUs). Given that mission-critical data are conveyed through these networks over a wireless medium, these networks are vulnerable to cyber-security threats from external adversaries \cite{v2x_attacks}. Therefore, next-generation V2X networks should possess the capability to detect security threats on their nodes and then initiate appropriate mitigation strategies. Thus, in order to assist detecting security threats on V2X networks, this paper proposes novel strategies to capture location information of its nodes via the data-flow logs.\looseness=-1

Inter-vehicle communication and direct communication between the vehicle and the RSUs may not be possible in V2X networks either due to transmit-power constraints or signalling blockage effects. In such a scenario, multi-hop communication assists in conveying the messages from the source vehicle to the RSU with the help of multiple intermediate nodes in an ad-hoc fashion. In order to learn the security vulnerabilities in such networks, the RSU should be able to remotely learn the state of the network, which contains the identity of packet forwarders, the sequence in which the packet is forwarded \cite{amogh}, \cite{suraj} and \cite{nodeembedding}, the spatial location of the vehicles that forwarded or originated the packet.  In particular, if the RSU knows the vehicles' locations, it can offer location-based security features. \looseness=-1

As vehicles use the Global Positioning System (GPS) for navigation purposes, they may embed the same when forwarding the packet. However, in a multi-hop communication setup, the participating vehicles may not want to share their exact location since the RSU and other vehicles can learn their exact location from the packet. Although incorporating an encrypted format of their GPS coordinates within the packet serves as a method to provide privacy on the vehicle's location from third-party observers, this approach increases the end-to-end delay for the packets and also discloses their precise whereabouts to the RSU, which may not be desirable.\looseness=-1 

From the above discussion, it is evident that the RSU would like to learn the precise locations of vehicles, whereas the forwarding vehicles might prefer not to disclose such information. Therefore, to strike a harmonious balance between the RSU's requirements and the privacy concerns of vehicles, we propose a solution where the RSU initially divides its coverage area into uniform-sized fragments and subsequently instructs the vehicles to incorporate the identities of their respective fragments when forwarding packets. To achieve this, vehicles can collaboratively determine an appropriate size of the fragment without disclosing their precise locations within the fragment. Consequently, the selection of fragment size emerges as a crucial parameter in safeguarding spatial provenance privacy.\looseness=-1

\begin{figure}
     \centering
    \includegraphics[trim={0 0 0 0},clip,scale=0.52]{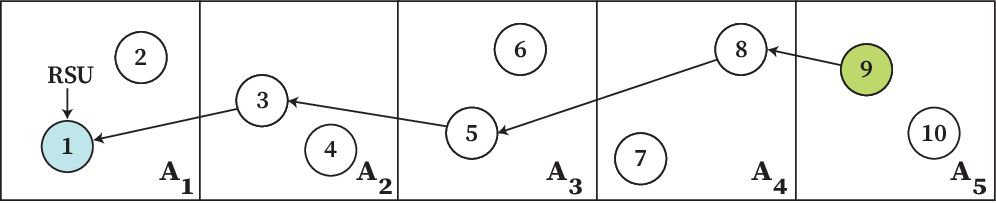}
        \setlength{\belowcaptionskip}{-15pt}   
        \caption{Depiction of a network wherein the nodes are distributed over a geographical area, and the RSU intends to learn spatial-provenance of the nodes through data-logs.}
        \label{fig: network Image}
\end{figure}

\subsection{Contributions}
We model the roads linearly in the context of vehicular networks, wherein a linear stretch of road in the communication range of the RSU is divided into linear segments of equal length, as exemplified in Fig. \ref{fig: network Image}. Subsequently, we propose the RSU to broadcast the segmentation information to vehicles as a dictionary. Finally, by using the dictionary received from the RSU, we propose the vehicles to learn and embed their segment's identity in the packet with the help of their GPS coordinates \cite{spatial}. This way, we show that the vehicles and the RSU can reach an amicable solution. To implement the above-proposed idea with certain communication constraints in the multi-hop setup, we use correlated linear Bloom filters (CLBF) to embed the spatial provenance information. In CLBF, two Bloom filters are used, wherein the first Bloom filter will be used to embed the edge identity, and the second Bloom filter will be used to embed the location of the vehicles. To help the vehicles choose the appropriate parameters of the Bloom filter, we derive an analytical expression for calculating the error rates in recovering spatial provenance from the CLBF. With the help of an optimization problem, we find a way to synthesize the Bloom filter parameters that can be used for embedding and recovering spatial provenance for any number of nodes and any number of fragments. To validate the impact of our method, we compare the analytical bound with the simulation results, which shows that the analytical expression offers near-optimal Bloom filter parameters.  

\section{Network model} \label{sec:network Model}
We consider a set of $N$ nodes, out of which $N-1$ nodes are mobile vehicles, and one node is the stationary RSU. Mobile nodes want to communicate with the RSU; however, they cannot communicate directly with the RSU due to limited power constraints. Therefore, nodes communicate with the RSU with the help of several intermediate nodes in a multi-hop manner.  To enhance the network's security, the RSU wants to learn the exact location of all the nodes in the network. However, due to privacy concerns, vehicles may not share their exact GPS coordinates with the RSU. Therefore, to balance the requirements of the RSU and the privacy of vehicles, RSU divides its coverage area into $r$ regions. \looseness=-1

The RSU models the vehicular network in a straight line and divides the road into smaller segments of equal length. The set of segments is denoted by $\Delta =\{A_1,A_2,\ldots,A_{r}\}$, as shown in Fig. \ref{fig: network Image}. We assume that the RSU is in area $A_1$, whereas the other mobile nodes are randomly distributed in $\Delta$. We also assume that the segments farther from the RSU get higher numbers of the index. For instance, the farthest segment in the coverage area is referred to as $A_{r}$. To assist the nodes in learning their regions, the RSU broadcasts a dictionary containing the boundaries of GPS locations mapped to different segments. Since the nodes are equipped with GPS, they will privately derive their segment identity (ID) using the broadcasted dictionary \cite{spatial}. The following section describes the routing protocol adopted for the multi-hop networks. \looseness=-1

\subsection{Routing Constraints} \label{sec:Routing Constraints}

We assume that routing protocols such as AODV \cite{aodv_original} are used in the network, which ensures minimum-hop delivery of the packet from the source to the destination. Therefore, we assume that the routing protocols prevent the formation of any loop or back-hop. For instance, a node in the area $A_i$ cannot send the packet to a node in the area $A_j$, where $i< j$. Another network constraint is that nodes can talk to the nodes in adjacent areas only (due to limited power constraint), i.e., if a node is present in area $A_3$, it can only send packets to area $A_2$ or $A_3$, and likewise, receive packets from segments $A_3$ or $A_4$.

With the above network protocol, the objective of the RSU is to learn the location of the nodes and the path travelled by the packet in the network. In V2V latency-constrained scenario, the constraint of the network model is to maintain constant packet size while ensuring node location confidentiality in a multi-hop network. Therefore, we use Bloom filters to embed the provenance information on the location as well as the path traversed by the packets.\looseness=-1

\section{Bloom Filter based Spatial-Provenance Recovery}\label{Bloom Filters}

To extract the path travelled by the packet and the locations of the vehicles that have forwarded the packet, we use Correlated Linear Bloom Filters (CLBF). CLBF is an extension of the standard Bloom filter that simultaneously performs insertions and queries from two sets. 
Bloom filter is selected for this task because of its properties such as fixed size, no need for additional encryption, guarantee of no false negatives, and constant time for insertions and queries. To explain the embedding operations with the Bloom filter, we denote the identity of the source node as $I_s$ and identities of the rest of the mobile nodes as $I_n$, where $n \in [1,~N-2]$. Also, the function $g: \{I_1, I_2,\ldots,I_{N-2}, I_s\} \rightarrow \Delta $ is used to map the identities of the nodes to the identities of the segment.\looseness=-1

\subsection{Embedding in Bloom filters}
It is assumed that the source node forwards the packet in a multi-hop manner to reach the RSU. Initially, the source node creates two empty Bloom filters: one for conveying the path information and another for conveying the segment information. These Bloom filters are referred to as edge ($BF_1$) and location ($BF_2$) Bloom filters, respectively, as shown in Fig. \ref{fig:CLBF}. Edge and location Bloom filters have sizes of $m_1$ and $m_2$ bits, respectively, with $k_1$ and $k_2$ Hash functions used for embedding the elements into the edge and location Bloom filters. At the start of the communication, the source node does not modify the edge Bloom filter; however, it embeds its node ID $I_s$ and its segment ID $g(I_s)$ into the location Bloom filter using $k_2$ number of Hash functions. Embedding in the location Bloom Filter ($BF_2$) will be done as follows: for a given $L$, where $L \in [1,k_2]$, $H_{L}(I_s, g(I_s), pid)$ randomly generates an index of location Bloom filter in $[1,m_2]$,
where $H_{L}$ denotes the $L^{th}$ Hash function, $pid$ denotes the packet ID, $g(I_s)$ denotes the segment ID. This packet is forwarded to the next node $I_n$, where $n \in [1,N-2]$ en-route to the RSU. Upon reception of the packet, the next node embeds the edge ID $(I_s, I_n, pid)$ in the edge Bloom filter using $k_1$ number of Hash functions. For more details on the edge embedding process, we refer the readers to \cite{amogh}. Whereas, in the location Bloom filter, the node $I_n$, $n \in [1,N-2]$,  embeds its segment ID using $H_{L}(I_n, g(I_n), pid)$ similar to that of the source node. The same process is followed at every node when the packet is en-route to the RSU.\looseness=-1


\begin{figure}
    \centering
    \includegraphics[scale=1.4]{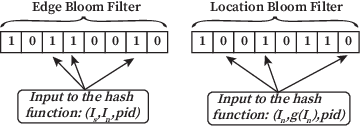}
    \vspace{0.4cm}
    \caption{Depiction of embedding in CLBF at the intermediate nodes which comprise edge and location Bloom filters of size $m_1 =8$ bits and $m_2=8$ bits, respectively. Here $k_1=k_2=3$ number of Hash functions are used for embedding an element in CLBF.}
    \label{fig:CLBF}
\end{figure}



\subsection{Recovery from Bloom filter} \label{sec:recovery}
After successful packet reception at the RSU, it intends to recovers the path travelled by the packets and the locations of the nodes that forwarded the packet. To do this, the RSU uses the edge Bloom filter and tests all possible edges to create a set of recovered edges. We assume that a Depth First Search (DFS) algorithm runs on the recovered set of edges to obtain the correct paths of the desired hop length. For each path, we pair each node of the chosen path with all possible segments, and these pairs are tested for their presence in $BF_2$. Using the set of all pairs that were recovered from $BF_{2}$, all possible paths and their locations are constructed, and finally the set of correct paths that satisfy the communication constraints described in Section \ref{sec:Routing Constraints} are extracted. In the next section, we will find the optimal network parameters to convey the spatial provenance. \looseness=-1


\section{On the choice of the Bloom filter parameters}
As the packet can travel only one path, the RSU should recover exactly one path. However, Bloom filters that save space and time in the embedding and query process may result in false positives, which could lead to multiple paths. The event of false positives, denoted by $E_{fp}'$, is defined as the scenario when more than one possible arrangement of nodes and segments satisfying the constraints in Section \ref{sec:Routing Constraints} is recovered using the Bloom filters. Towards optimizing the performance, we need to choose the parameters $m_1,~ m_2,~k_1$ and $k_2$ so as to minimize the average probability of false positives. Therefore, we propose the following problem statement.\looseness=-1


{\begin{problem_stmt}{}
Given $N$, a geographical bifurcation of the area $\Delta$, and given CLBF size $m$ solve:
\begin{IEEEeqnarray}{rCl}
{k_1}^{*}, {k_2}^{*} =\arg \min_{\{k_1,k_2\}} \Pr(E_{fp}'),
\end{IEEEeqnarray}
 s.t. $1\leq k_1\leq m_1$, $1\leq k_2\leq m_2$ and $m_1 + m_2 = m$.
\end{problem_stmt}}

Since the recovery process occurs in phases, i.e., restoring the path from $BF_1$ and subsequently recovering locations from $BF_2$ (see Section \ref{sec:recovery}), we can express the overall false positive  probability as 
\begin{equation*}
\Pr(E_{fp}') = \Pr(fp_1) + (1-\Pr(fp_1)) \times \Pr(fp_2)
\end{equation*}
where $\Pr(fp_1)$ denotes the average false positive probability of $BF_1$, and $\Pr(fp_2)$ denotes the average false positive probability of $BF_2$ conditioned on no false positive event from $BF_{1}$. If the size of $BF_1$ is large enough, and the number of Hash functions in it is already optimized, then $\Pr(fp_1) \approx 0$. As a result, we can approximate $\Pr(E_{fp}') \approx \Pr(fp_2)$, conditioned that probability $\Pr(fp_1) \approx 0$.\looseness=-1

Using the assumption stated above, we propose a simplified, yet non-trivial problem statement on optimizing the parameters of $BF_{2}$.\looseness=-1

{\begin{problem_stmt}{}\label{problem2}
Given $N$, a geographical bifurcation of the area $\Delta$, and a given distribution of the nodes, solve:
\begin{IEEEeqnarray}{rCl}
{k_2}^{*} =\arg \min_{\{k_2\}} \Pr(E_{fp}),
\end{IEEEeqnarray}
 s.t. $1\leq k_2\leq m_2$.
\end{problem_stmt}}

\subsection{Optimization of Bloom Filter Parameters}

In order to optimize the parameters defined in Problem \ref{problem2}, we formally define the false positive event of $BF_{2}$.

\begin{definition}
A false positive event, denoted by $E_{fp}$, occurs when a single path is recovered from $BF_1$, however, more than one location sequence satisfying the routing constraints in Section \ref{sec:Routing Constraints} are recovered from $BF_2$.
\end{definition}

Let us assume that a single path of hop length $h$ is already recovered from $BF_1$, and the nodes of this path have been paired with various segment IDs to recover the locations of the nodes. Given that the false positive probability of $BF_{2}$ is a function of $m_2$, $k_2$, $h$ and $\Delta$, we need to derive an expression as a function of these parameters.

Given that $k_{2}$ Hash functions are used by every node, the number of bits lit in $BF_{2}$ through the journey of the packet is a random variable with a minimum value of $1$ and a maximum value of $min(m_2,k_2h)$. Thus, using $\alpha$ to denote the number of bits that have been lit in $BF_2$ by all the forwarding nodes, the false positive probability is given in the following proposition.

\begin{proposition} \label{prop:Pr_Efp} The probability of false positive event, denoted by, $\Pr(E_{fp})$ can be written using the Bayes' rule as:
\begin{IEEEeqnarray}{rCl}
    \Pr(E_{fp}) = \sum\limits_{\alpha=1}^{\min(m_{2},k_{2}h)}\Pr(E_{fp}|\alpha)\Pr(\alpha),
\end{IEEEeqnarray}
where $\Pr(E_{fp}|\alpha)$ denotes the false positive probability conditioned on $\alpha$ bits lit in $BF_{2}$, and $\Pr(\alpha)$ denotes the probability of $\alpha$ bits lit in the $BF_2$, which is given by \cite{suraj}:
\begin{IEEEeqnarray}{rCl}
\label{eqn:Pr_alpha}
\Pr(\alpha)=\frac{{m_2 \choose \alpha}  \sum\limits_{\gamma=0}^{\alpha}(-1)^{\gamma} {{\alpha}\choose{\gamma}}(\alpha-\gamma)^{k_2 h} }{m_2^{k_2 h}}.
    \end{IEEEeqnarray}
\end{proposition}

Before we present our method to compute $\Pr(E_{fp}|\alpha)$, we discuss preliminaries related to the false positive events associated with Bloom filters. 

\begin{definition}\label{def:collision} A collision event in a Bloom filter is defined as an event when we verify an object from the Bloom filter that was not initially embedded in it.
\end{definition}

Given a Bloom filter of size $m$ bits with $\alpha$ bits already lit in it, the probability of a collision event is given by $(\alpha/m)^{k}$, where $k$ is the number of Hash functions. With respect to $BF_{2}$, when we verify the segment ID of a node, the corresponding probability of a collision event, denoted by $p_1$, is given as $p_1 = (\alpha/m_2)^{k_2}$. Consequently, the probability of the non-collision event, denoted by $p_2$, is given as, $p_2 = 1 - p_1$. Note that the above discussed collision events are crucial in enumerating the false positive events in $BF_{2}$. To enumerate the false positive events in $BF_2$, we define a valid segment sequence as the set of ordered locations that satisfy the communication constraints, as defined in Section \ref{sec:Routing Constraints}. In simple terms, a valid sequence of segments is a sequence from $\Delta$ such that the indexes of segments are in non-decreasing order with the condition that the difference between consecutive indexes is not more than one. For instance, for a hop-length of $h = 4$ and $|\Delta| = 6$, a valid sequence is $\{A_{1}, A_{2}, A_{2}, A_{3}\}$, and not $\{A_{1}, A_{2}, A_{4}, A_{5}\}$. In general, denoting the set of all valid sequences of segments as $\mathcal{P}$, the following proposition presents the result on the cardinality of $\mathcal{P}$ as a function of $|\Delta|$ and $h$.

\begin{lemma} \label{lemma:path}
With $\mathcal{P}$ denoting the set of all valid sequences of segments, we have $|\mathcal{P}| = \sum\limits_{i=1}^{|\Delta|}\binom{h}{i-1}.$    
\end{lemma}
\begin{proof}
To count the number of valid segment sequences, we focus on the required order on subscripts of the segment index $A_{i}$, for $i \in [1, |\Delta|]$. As the number of nodes in a $h$-hop path is $h+1$, we look at the sequences of non-decreasing numbers from $1$ to $h+1$. Due to the communication constraints, the indices in a valid sequence range from $1$ to $i$, for $i \in [1, |\Delta|]$, although we may have $|\Delta| > h$. Furthermore, let $b_l$ denote the number of nodes present in the segment with index $l$, where $1 \leq l \leq i$. It is clear that $b_{l} \geq 1$ for each $l$ in the above range. This also implies that for each $i$, the sum of $b_{l}$'s over all $l$ must be $h+1$. Thus, the number of valid sequences of length $h+1$ with values from $\{1, 2, \ldots, i\}$ can be counted as the distinct number of unique tuples of $i$ positive integers, ($a_1, a_2,\ldots, a_i), ~a_j \in \mathbb{N},~ j\in[1,i]$, where $\sum \limits_{k=1}^i a_k = h+1$. To count the latter, the problem can be modelled as the distribution of $h+1$ identical objects in $i$ different bins, such that no bin is empty. To solve this problem, consider $h+1$ objects arranged in a line, and note that there are $h$ gaps between these objects. If we select $i-1$ gaps out of $h$ gaps, we can form $i$ different bins. Thus, the total number of possible ways is $\binom{h}{i-1}$. Overall, since $i$ can vary from $1$ to $|\Delta|$, we have $|\mathcal{P}| = \sum\limits_{i=1}^{|\Delta|}\binom{h}{i-1}.$ This completes the proof.
\end{proof} 

Based on Section \ref{Bloom Filters}, during recovery of the node-segment pairs from the Bloom filter $BF_2$, we check all possible combinations of node-segment pairs in the Bloom filter. Therefore, the RSU verifies $h|\Delta|$ pairs in the Bloom filter, out of which only $h$ were actually embedded. We define the set of remaining pairs as false pairs, as formally defined below.
\begin{definition}
The set of node-segment pairs that were not originally embedded in $BF_2$, however, are checked for membership in $BF_2$ at the RSU, is called the set of false pairs, denoted by $\mathcal{F}$, where $|\mathcal{F}| = h(|\Delta|-1)$.
\end{definition}
\begin{definition}
The set of node-segment pairs recovered from $BF_2$, not on the original segment sequence, is called extra recoveries, denoted by $\mathcal{R}$. Note that $\mathcal{R} \subseteq \mathcal{F}$.
\end{definition}

To count the false positive events, suppose that the nodes are distributed in the network with the location sequence lying in the set $\mathcal{P}$. For that given segment sequence, false positive event occurs when either one of the other sequences in $\mathcal{P}$ is recovered from the Bloom filter. Since other sequences in $\mathcal{P}$ are recovered through extra recoveries, we need to relate extra recoveries with the false positive events. During the recovery process from $BF_{2}$, the number of extra recoveries is a random variable. Therefore, we need to consider all possible cases of extra recoveries such that $|\mathcal{R}| \leq |\mathcal{F}|$. Assuming $\mathcal{P} = \{\mathbf{p}_{1}, \mathbf{p}_{2}, \ldots, \mathbf{p}_{|\mathcal{P}|}\}$, suppose that $\mathbf{p}_{i}$ for some $i \in [1, |\mathcal{P}|]$, is the true sequence of segments. If the Bloom filter is lit based on this sequence $\mathbf{p}_{i}$ from the $h$ nodes, let $C_{ij}$ represent the number of sequences in $\mathcal{P}\backslash \{\mathbf{p}_{i}\}$ which can be recovered from the Bloom filter when $j = |\mathcal{R}|$ extra recoveries are lit in the Bloom filter. If $C_{ij}$ can be computed for all $j \in [1, |\mathcal{F}|]$, then the false positive probability for the sequence $\mathbf{p}_{i}$ can be computed as
\begin{equation*}
\Pr(E_{fp}|\alpha, \mathbf{p}_{i}) = \sum_{j=1}^{h(|\Delta|-1)} p_1^j\times p_2^{h(|\Delta|-1)-j}\times C_{ij}.
\end{equation*}
Thus, the false positive probability averaged over all possible true segment sequences can be computed as
\begin{equation*}
\Pr(E_{fp}|\alpha) = \sum_{i = 1}^{|\mathcal{P}|}\Pr(E_{fp}|\alpha, \mathbf{p}_{i})\Pr(\mathbf{p}_{i})
\end{equation*}
\begin{equation}
\label{overall_fp}
 = \frac{1}{|\mathcal{P}|}\sum_{j=1}^{h(|\Delta|-1)} p_1^j\times p_2^{h(|\Delta|-1)-j}\times (\sum_{i = 1}^{|\mathcal{P}|}C_{ij})
\end{equation}
where uniform distribution is assumed on $\mathbf{p}_{i}$. Using the above expression, given the parameters of Bloom filters, the expression for $\Pr(E_{fp}|\alpha)$ can be computed provided $C_{j} \triangleq \sum_{i = 1}^{|\mathcal{P}|}C_{ij}$ is computed. In the next subsection, we present a way to compute $\{C_{j}, j = 1, 2, \ldots, h(|\Delta|-1)\}$. Once we present this method, we can use its outcome and apply them in \eqref{overall_fp}. 

\subsection{On computing $\{C_{j}\}$}

First, we discuss our method to compute $C_{1}$, and then subsequently provide a method to compute $\{C_{j}, j > 1\}$
Let us now consider a case where we get only one extra recovery from $BF_2$, i.e., $|\mathcal{R}| = 1$. Note that we can get one extra recovery in $|\mathcal{F}|$ possible ways, and not all these $|\mathcal{F}|$ pairs will cause a false positive event. For example, consider a segment sequence $(I_{1},A_1)->(I_{2},A_2)->(I_{3},A_2)->(I_{4},A_3)$. If $\mathcal{R} = \{(I_{4},A_1)\}$, this will not cause a false positive event, as substituting the extra recovery in the original path will violate the node communication constraint in Section \ref{sec:Routing Constraints}. However, if $\mathcal{R} = \{(I_{4},A_2)\}$, then when we replace this with $(I_{4},A_3)$, this will cause a false positive event. For a given segment sequence, let the number of false pairs that results in a false-positive event when appearing as single recovery is denoted by $J$. With that, we have the following bound.

\begin{definition} \label{prop:limit j}
The minimum and maximum value of $J$ is $1$ and $2|\Delta|-2$, respectively. 
\end{definition}
\begin{proof}
The minimum value of $J$ is 1 since $J = 0$ cannot result in false positives, trivially. However, for the maximum value, consider a segment sequence with all the segments up to $|\Delta|$. Excluding the extreme segment IDs, i.e. $A_1$ and $A_{r}$, for all the other segments, we may either increase or decrease the segment index to obtain another possible valid segment sequence, potentially. Thus, we can have $2(|\Delta|-2)$ as the maximum value of $J$ by considering the false pairs that replace the intermediate nodes. For each of the extreme nodes, the possibility is only one, i.e., nodes in region $A_1$ can only increase to create an alternate valid segment sequence, whereas the nodes in region $A_{r}$ can only decrease. Therefore, the maximum limit on the value of $J$ including all the nodes is $2(|\Delta-2|)+2 = 2|\Delta|-2$.
\end{proof}
For different valid segment sequences in the set $\mathcal{P}$, the value of $J$ can be different. We introduce a parameter, referred to as $f_J$, which is defined as the number of segment sequences in $\mathcal{P}$, that have exactly $J$ false pairs resulting in false positive events
when appearing as extra single recovery. Once $f_{J}$ is estimated for each $J$, then it is easy to verify that 
$C_{1} = \sum_{J = 1}^{2|\Delta| - 2} f_{J}.$

In the rest of this section, we present a method to compute $\{C_{j}, j > 1\}$. Towards that direction, we consider the events when $j$ false pairs in $\mathcal{F}$ are jointly recovered from the Bloom filter, and then they together result in a false positive event. We refer to such events as extra $j$ recoveries.

\begin{proposition}\label{prop:value extension}
Given $J$ for a segment sequence, the number of sets of false pairs such that $|\mathcal{R}| = j$, which will cause a false positive event is given by:
\begin{eqnarray}\label{eq:j}
\sum_{l=1}^{J} \binom{h(|\Delta|-1) - l}{j-1}.
\end{eqnarray}
\end{proposition}

\begin{proof}
We derive this expression by sequentially selecting one of the $J$ pairs, and then selecting the remaining, $j-1$ pairs from the remaining false pairs of $\mathcal{F}$, in order to construct the set $\mathcal{R}$. Subsequently, two pairs from the $J$ pairs are fixed, and then the rest of the $j-2$, pairs can be picked from the remaining false pairs of $\mathcal{F}$, in order to construct the set $\mathcal{R}$. This process is repeated up to the case of fixing $J$ pairs to fill $\mathcal{R}$, and then picking the rest of the false pairs from $\mathcal{F}$. 
\end{proof}

Using the above result, we are now ready to have an expression for $\{C_{j}, j > 1\}$, given by
\begin{IEEEeqnarray}{rCl}
\label{eq:cj_expression}
C_j = \sum_{J=1}^{2|\Delta|-2} f_J\times\sum_{l=1}^{J}\binom{h(|\Delta|-1) - l}{j-1}.
\end{IEEEeqnarray}
From the above expression, it is clear that as long as we have $f_{J}$, we can compute $\{C_{j}\}$. Furthermore, by plugging the above values into \eqref{overall_fp}, we obtain the expression for average false positive probability. In the next section, we present our method to compute $\{f_{J}, J = 1, 2, \ldots, 2|\Delta| - 2\}$.

\begin{figure*}
\begin{small}
\begin{eqnarray}
\label{eqn:EDelta}
f^{e}_{J,|\Delta|} = \left[ \sum\limits_{i=0}^{J/2-1}\binom{J/2-1}{i} \times \binom{|\Delta| -J/2-1}{i}\right] \left[ \sum\limits_{i=0}^{J/2}\binom{J/2}{i} \times \binom{h-|\Delta| -J/2+1}{i}\right] \mbox{when } J \mbox{ is even},
\\
\label{eqn:EDeltaDash}
f^s_{J,\delta} = \left[ \sum\limits_{i=0}^{J/2-1}\binom{J/2-1}{i} \times \binom{\delta -J/2-1}{i}\right] \left[ \sum\limits_{i=0}^{J/2-1}\binom{J/2-1}{i}  \binom{h-\delta-J/2+1}{i}\right] \mbox{when } J \mbox{ is even}.
\end{eqnarray}
\end{small}
\hrulefill
\end{figure*}
\

\subsection{On computing $\{f_{J}\}$}
\label{sec:pattern}

Recall that $f_J$ represents the number of valid segment sequences in $\mathcal{P}$ that experience false positive events upon recovery of either of the $J$ false pairs as one extra recovery. To compute $f_{J}$, we enumerate the set of valid segment sequences based on the number of segments they contain, and then compute $f_{J}$ for each of such cases. In particular, we denote $f_{J,\delta}$ as the number of valid segment sequences containing only $\delta$ segments that experience false positive events upon recovery either of the $J$ false pairs as one extra recovery. With that we have $f_J = \sum_{\delta=1}^{|\Delta|}f_{J,\delta}$. While computing $f_{J,\delta}$ for each $\delta$, we notice that the values of $f_{J,\delta}$, for $1 \leq \delta \leq |\Delta|-1$ has a unique pattern that is deterministic. Also, the values for $f_{J,|\Delta|}$ has a pattern that is different from that of smaller values of $\delta$. Generalizing this pattern through regression, we propose to obtain $f_{J}$ as
\begin{IEEEeqnarray}{rCl}\label{eqn:f_j expanded}
    f_J = f^e_{2 \lfloor{J/2}\rfloor,|\Delta|} + \sum \limits_{\delta=1}^{|\Delta|-1} f^s_{2 \lfloor {J/2} \rfloor,\delta},
\end{IEEEeqnarray}
where the terms on the RHS of the above equation can be calculated using \eqref{eqn:EDelta} and \eqref{eqn:EDeltaDash} respectively. Finally, we use \eqref{eqn:f_j expanded} in \eqref{eq:cj_expression}, and in turn use them in \eqref{overall_fp} to get the expressions for average false positive probability.

\section{Experimental Results}


\begin{figure}
    \centering
    \includegraphics[width = 0.45\textwidth]{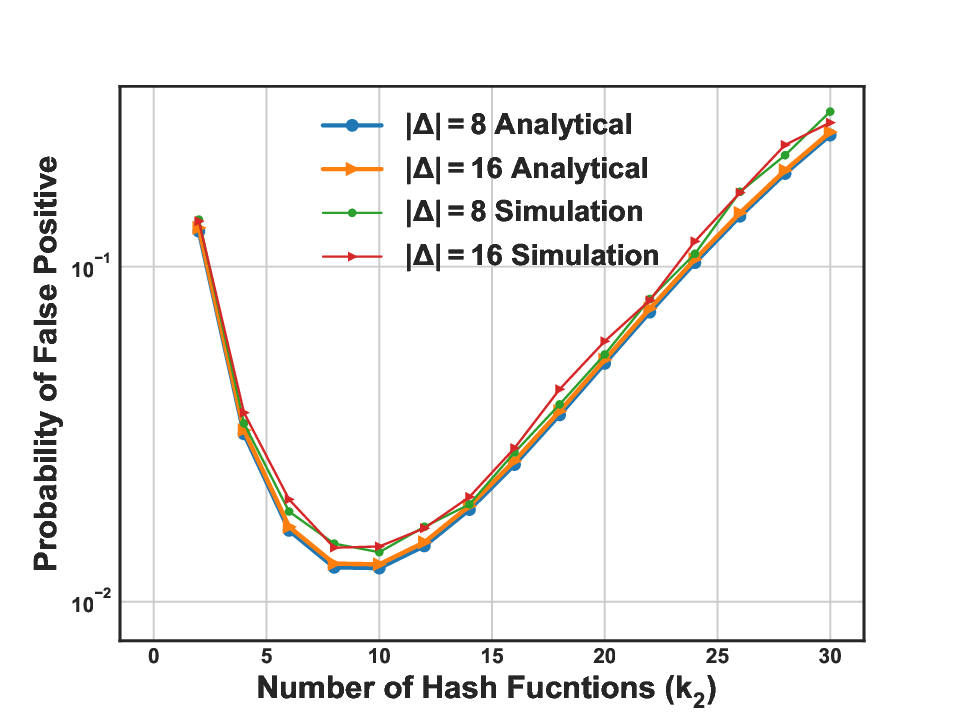}
    \caption{False positive probability using the analytical bound and the simulation results for varying values of $k_2$ on a network of $N = 16$, with $|\Delta| = 16$ and $|\Delta| =8$.}
    \label{fig:mix8-16}
\end{figure}

\begin{figure}
    \centering
    \includegraphics[width = 0.45\textwidth]{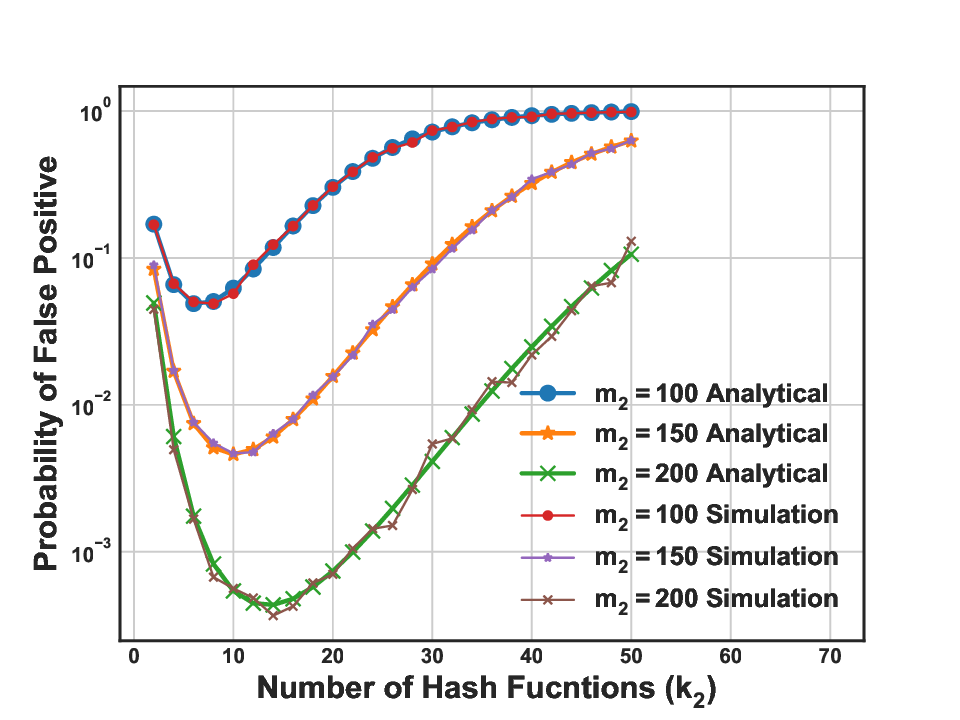}
    \caption{False positive probability using the analytical bound and the simulation results with varying $m_2$ on a network of $N = 11$ with $|\Delta| = 15$ and $h=10$.}
    \label{fig:m-varying}
\end{figure}

\begin{figure}
    \centering
    \includegraphics[width = 0.45\textwidth]{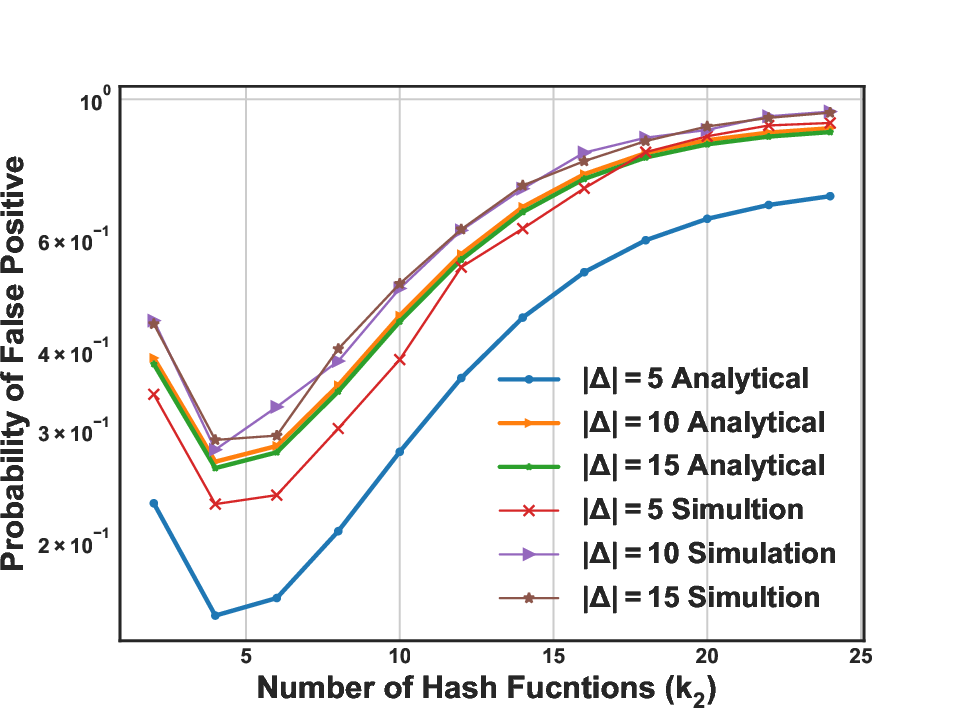}
    \caption{False positive probability using the analytical bound and the simulation results with varying $|\Delta|$ and constant $m_2=100, ~h=15$.}
    \label{fig:Delta_varying}
\end{figure}

In this section, we present experimental results to showcase the efficacy of our solution. To verify the correctness of the analytical expressions on the false positive probability, we compare them with the false positive probability generated empirically through simulation results. We simulated a network with $N = 16$, $h = 15$, $|\Delta| = 8$ where each segment has two nodes. Bloom filter of size $m_2=200$ bits is used, where the number of Hash functions is varied from $2 \leq k_2 \leq m_2$. For this setup, in Fig. \ref{fig:mix8-16}, we plot the resultant false positive probability obtained through analytical expressions. We repeated the same using simulation results for the same parameters. From Fig. \ref{fig:mix8-16}, we observe that the curve of the analytical expression acts as a good approximation to that generated via simulation results, and the minima of both the plots coincide. Note that the minima of the curves gives the optimal number of Hash functions, which will be used for embedding spatial provenance information in practice. In Fig. \ref{fig:m-varying}, we present similar results by varying the value of $m_2$ for a 10-hop network of $N=11$ nodes spread across $|\Delta|=15$. We observe that as the size of $m_2$ increases, the false positive probability decreases. This implies that with larger packet size, localization accuracy at the RSU will improve subject to a given choice of $|\Delta|$ as agreed by the nodes. Similarly, Fig. \ref{fig:Delta_varying} also shows that as the RSU needs to learn higher resolution of location for a given packet size, the accuracy of localization reduces. Thus, the only way to learn higher resolution of location with high accuracy is to increase the packet size.


\section{Discussion}

This study focused on learning the spatial provenance in multi-hop networks with the help of correlated Bloom filters. The main emphasis was to jointly handle the privacy of the nodes as well as the localization requirements of the RSU. For future work, an interesting direction is to design the Bloom filter parameters when the communication constraints imposed on the nodes in this work are relaxed.  


\section*{Acknowledgement}
This work was supported by the project titled ``Development of Network Provenance Techniques for Monitoring Wireless Networks" under Contract for Acquisition of Research Services (CARS) from the Scientific Analysis Group (SAG), DRDO, New Delhi, India.\looseness=-1

\bibliographystyle{IEEEtran} 
\bibliography{refs}

\end{document}